\newif\ifonecol

\ifonecol
\documentclass[peerreview,12pt]{IEEETran}
\else
\documentclass[journal]{IEEEtran}
\fi

\usepackage[dvips]{graphicx}
\usepackage{amsmath,amssymb}
\usepackage{graphicx}
\usepackage{placeins}
\usepackage{wrapfig}
\usepackage{verbatim}
\usepackage{flushend}
\usepackage{subcaption}
\usepackage{tikz}
\usepackage{amsthm}
\usepackage{varioref}
\usetikzlibrary{calc}
\usepackage[ruled,vlined,boxed,algonl,linesnumbered]{algorithm2e}

\theoremstyle{theorem}
\newtheorem{theorem}{Theorem}
\newtheorem{corollary}{Corollary}
\newtheorem{lemma}{Lemma}

\theoremstyle{definition}
\newtheorem{definition}{Definition}
\newtheorem{remark}{Remark}
\theoremstyle{remark}
\newtheorem{example}{Example}

\makeatletter
\newcommand*\bigcdot{\mathpalette\bigcdot@{.5}}
\newcommand*\bigcdot@[2]{\mathbin{\vcenter{\hbox{\scalebox{#2}{$\m@th#1\bullet$}}}}}
\makeatother

\newcommand{\bnull}{\mathbf{0}}

\newcommand{\bF}{\mathbb{F}}

\newcommand{\bS}{\mathbb{S}}
\newcommand{\bfT}{\mathbf{T}}

\newcommand{\mA}{\mathcal{A}}
\newcommand{\mB}{\mathcal{B}}
\newcommand{\mC}{\mathcal{C}}
\newcommand{\mD}{\mathcal{D}}
\newcommand{\mE}{\mathcal{E}}
\newcommand{\mF}{\mathcal{F}}
\newcommand{\mI}{\mathcal{I}}

\newcommand{\mP}{\mathcal{P}}

\newcommand{\mS}{\mathcal{S}}
\newcommand{\mT}{\mathcal{T}}
\newcommand{\mU}{\mathcal{U}}
\newcommand{\mW}{\mathcal{W}}
\newcommand{\mX}{\mathcal{X}}
\newcommand{\mY}{\mathcal{Y}}

\newcommand{\ove}{\overline{\mE}}

\newcommand{\vp}{\varphi}

\newcommand{\ceil}[1]{\left\lceil{#1}\right\rceil}

\newcommand{\floor}[1]{\left\lfloor{#1}\right\rfloor}

\newcommand{\la}{\langle}

\newcommand{\ra}{\rangle}
\newcommand{\set}[1]{\left\{ #1 \right\}}
\newcommand{\hull}[1]{\left\langle{#1}\right\rangle}
\newcommand{\wt}{\mathbf{wt}}

\DeclareMathOperator{\cs}{cs}

\DeclareMathOperator{\supp}{supp}

\begin{document}

\title{On Distance Properties of Convolutional Polar Codes}
\author{Ruslan Morozov, \IEEEmembership{Member,~IEEE}, Peter Trifonov, \IEEEmembership{Member,~IEEE}
\thanks{The authors are with the Saint Petersburg Polytechnic University, Russia. E-mail: \{rmorozov, petert\}@dcn.icc.spbstu.ru}}%
\sloppy

\maketitle
\begin{abstract}
A lower bound on minimum distance of convolutional polar codes is provided. The bound is obtained from the minimum weight of generalized cosets of the codes generated by bottom rows of the polarizing matrix.
Moreover, a construction of convolutional polar subcodes is proposed, which provides improved performance under successive cancellation list decoding. For sufficiently large list size, the decoding complexity of convolutional polar subcodes appears to be lower compared to Arikan polar subcodes with the same performance.
The error probability of successive cancellation list decoding of convolutional polar subcodes is lower than that of Arikan polar subcodes with the same list size. 
\end{abstract}

\section{Introduction}
In this paper we consider codes that were firstly introduced as branching-MERA codes \cite{ferris2013branching} and then as convolutional polar codes (CvPCs) \cite{ferris2017convolutional} by A.~J.~Ferris, C.~Hirche and D.~Poulin.
These codes were shown to provide substantially better performance under successive cancellation (SC) decoding compared to classical  polar codes \cite{arikan2009channel}.
In \cite{ferris2017convolutional} both open-boundary and periodic-boundary CvPCs are presented, in this paper by CvPCs we always mean open-boundary CvPCs.
In \cite{morozov2018efficient} the efficient min-sum implementation of SC decoding is presented for CvPCs, which requires one to perform only comparisons and additions and can be easily extended to the case of SC list (SCL) decoding.
Other implementations of SCL decoding for CvPCs are presented in \cite{saber2018convolutional, prinz2018successive}.

Classical polar codes provide quite poor performance under SCL decoding due to very low minimum distance, which scales as $O(\sqrt{n})$ \cite{hussami2009performance}.
Although the minimum distance of a polar code can be found simply, the problem of computing minimum distance of an arbitrary linear code is NP-complete. However, for moderate-length codes minimum distance can be obtained by method presented in \cite{canteaut98new}.

The generator matrix of a CvPC consists of rows of $n\times n$ non-singular matrix $Q^{(n)}$, called convolutional polarizing transformation (CvPT).  In this paper we derive a tight lower bound on the minimum distance of CvPCs, based on computing the minimum weight of a coset, given by the $i$-th row of CvPT, of a linear code, generated by the last $n-i-1$ rows of CvPT.
The weight enumerator polynomial of such coset can be expressed as $A_i(x)-A_{i+1}(x)$, where $A_i(x)$ is a weight spectrum of code generated by the last $n-i$ rows of matrix $Q^{(n)}$.
In the case of polar codes, an efficient method for approximate enumerator evaluation is available \cite{li2012adaptive}.
However, for convolutional polar codes there are no methods for evaluation of coset enumerator.

The minimum distance of CvPCs appears to be of the same order as in the case of classical polar codes.
However, by generalizing the construction of randomized polar subcodes \cite{trifonov2017randomized} to the case of CvPC, we obtain convolutional polar subcodes (CvPSs) with reduced error coefficient, which provide superior performance under SCL decoding, compared to polar subcodes.

The paper is organized as follows. In Section~\ref{s:bg} we introduce representation of linear block codes, which is natural for the cases of Arikan and convolutional polar codes.
The concepts of generalized cosets and recoverable vectors are introduced in Section~\ref{s:lower} and are used to obtain a lower bound on the minimum distance of linear block codes.
An efficient algorithm for computing the lower bound in the case of CvPC is provided in Section~\ref{s:mindist}. This algorithm is aimed to explore some properties of low-weight codewords of CvPC.
These properties are used for a construction of  convolutional polar subcodes, which is proposed in Section~\ref{s:subcodes}.
The performance of the proposed code construction is presented in Section~\ref{s:perf}.


\section{Background}
\label{s:bg}
\subsection{Notations}
The following notations are used throughout the paper. $\bF$ denotes the Galois field of two elements. For integer $n$ we denote $[n]=\{0,1,\ldots n-1\}$. For vector $a$ symbol $a_b^c=(a_b,a_{b+1},\ldots, a_c)$.
For two vectors $a$ and $b$ we denote their concatenation by $(a,b)$. 
For $m \times\ n$ matrix  $A$ and sets $\mX \subseteq [m]$, $\mY \subseteq [n],$ by $A_{\mX,\mY}$ we denote the submatrix of $A$ with rows with indices from set $\mX$ and columns with indices from set $\mY$, indexing of rows and columns starts with zero.
Similar notations are applied to vectors as well.
If  $\mX=*$ or $\mY=*$, this means that all rows or all columns of the original matrix are in the submatrix. 
Furthermore, $A_{\overline \mX,\overline \mY}$ denotes submatrix of $A$ consisting  of rows and columns with indices that are not in $\mX$ and $\mY$, respectively.
The vector of $i$ zeroes is denoted by $\bnull^i$, or just by $\bnull$ if $i$ is clear from the context.


\subsection{A Representation of a Linear Block Code and Successive Cancellation Decoding}
Consider binary linear block code in the form
\begin{align}
\set{u_0^{n-1}G^{(n)}\big\vert u_{\mI}\in \bF^k, u_{\mF}=\bnull}, \mI \subseteq [n], |\mI|=k,
\label{eq:polarlike}
\end{align}
where $G^{(n)}$ is an $n \times n$ non-singular binary matrix, $\mI$ is called information set and $\mF=[n]\setminus\mI$ is called frozen set.
The generator matrix of such code is $G^{(n)}_{\mI,*}$.
Note that any $(n,k)$ linear code with generator matrix $G$ can be expressed as in \eqref{eq:polarlike} with $G^{(n)}$, such that $G=G^{(n)}_{\mI,*}$ for some $\mI\subseteq [n]$.
For example, classical polar codes \cite{arikan2009channel} have $G^{(n)}=F^{\otimes m}$ for $n=2^m$.

For such code representation, the successive cancellation (SC) decoding method can be defined. Consider transmission of codeword $c_0^{n-1}=u_0^{n-1}G^{(n)}$ through binary-input memoryless channel $\mW:\bF\to\mY$.
Let $y_0^{n-1}$ be the output of this channel.
After demodulation, the probabilities $W(c_i|y_i)=\mW(y_i|c_i)/\left(\mW(y_i|0)+\mW(y_i|1)\right)$ for $c_i \in \bF$ are provided to the decoding algorithm.
Given the prior hard decisions $\hat u_0\ldots \hat u_{\vp-1}$, at phase $\vp$ the SC decoding algorithm calculates probabilities $W^{(\vp)}_n(\hat u_0^{\vp-1},u_\vp | y_0^{n-1})$, defined as
\begin{align}
W^{(\vp)}_n(u_0^{\vp}|y_0^{n-1})=\sum_{u_{\vp+1}^{n-1}\in\bF^{n-\vp-1}}W^n(u_0^{n-1}G^{(n)}|y_0^{n-1}),
\label{eq:wdef}
\end{align}
where $W^n(c_0^{n-1}|y_0^{n-1})=\prod_{i=0}^{n-1}W(c_i|y_i)$.
The channels $W^{(\vp)}_n:\mY \to \bF^{\vp+1}$ are called \textit{bit subchannels}. Then, the hard decision on  $u_{\vp}$ is made by
\begin{align*}
\hat u_\vp=\begin{cases}
0, &\vp\in\mF \\
\arg \displaystyle\max_{u_{\vp}\in\bF}W^{(\vp)}_n(\hat u_0^{\vp-1},u_\vp|y_0^{n-1}) , &\vp \notin \mF.
\end{cases}
\end{align*}

The SC decoding can be defined for any linear code, if an efficient  method for computing  $W^{(\vp)}_n(u_0^{\vp}|y_0^{n-1})$ is available. However, SC decoding can provide reasonable performance only for codes with $G^{(n)}$, such that the capacities of bit subchannels $W^{(\vp)}_n$ polarize, i.e. converge to $0$ or $1$ with $n \to \infty$.




\subsection{Convolutional Polar Codes}
Convolutional polar codes \cite{ferris2017convolutional} (CvPCs) are a family of linear block codes, for which $G^{(n)}$, $n=2^m$, is equal to the matrix of convolutional polarizing transformation (CvPT) $Q^{(n)}$, such that 
\begin{align}
Q^{(n)}=\left(X^{(n)}Q^{(n/2)},Z^{(n)}Q^{(n/2)}\right),
\label{eq:cptxz}
\end{align}
where $Q^{(1)}=(1)$, $X^{(l)}$ and $Z^{(l)}$ are $l\times l/2$ matrices, defined for even $l$ as
\begin{align}
\label{eq:xdef}
&X^{(l)}_{i,j} =\begin{cases} 1, & \text{if } 2j\leq i \leq 2j+2\\
0, & \text{otherwise}
\end{cases}\\
\label{eq:zdef}
&Z^{(l)}_{i,j} =\begin{cases} 1, & \text{if } 2j< i \leq 2j+2\\
0, & \text{otherwise}
\end{cases}
\end{align}
For example, 
$X^{(4)}=\begin{pmatrix}1110\\0011\end{pmatrix}^T$,
$Z^{(4)}=\begin{pmatrix}0110\\0001\end{pmatrix}^T$.
Expansion \eqref{eq:cptxz} corresponds to one \textit{layer} of CvPT.
In Fig.~\ref{fig:cpt}, the $m$-th layer of CvPT is a mapping of vector $u_0^{n-1}$ to vectors $x_0^{n/2-1}=u_0^{n-1}X^{(n)}$ and $z_0^{n/2-1}=u_0^{n-1}Z^{(n)}$.

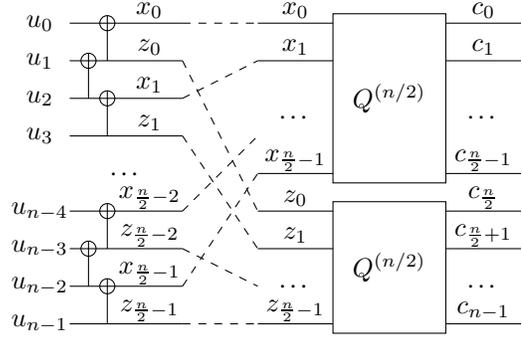
\begin{figure}
\centering
\begin{tikzpicture}[x=0.5cm,y=0.5cm]

\def\r{0.2}
\def\g{0.35}


\foreach \x in {1,...,4,6,7,...,9}{
  \draw (1,\x)--(4,\x);
}
\foreach \x in {1,3,4,5,8,9}{
\draw (6,\x)--(8,\x);
}
\foreach \x in {1,3,4,5,8,9}{
\draw (11,\x)--(13,\x);
}


\foreach \x in {2,7}{
\draw (1.5,\x)--(1.5,\x+1+\r);
\draw (1.5,\x+1) circle (\r);
}

\foreach \x in {1,3,8,6}{
\draw (2,\x)--(2,\x+1+\r);
\draw (2,\x+1) circle (\r);
}


\draw[dashed] (4,9)--(6,9);
\draw[dashed] (4,8)--(6,4);
\draw[dashed] (4,7)--(6,8);
\draw[dashed] (4,6)--(6,3);

\draw[dashed] (4,2)--(6,5);
\draw[dashed] (4,1)--(6,1);
\draw[dashed] (4,3)--(6,2);
\draw[dashed] (4,4)--(6,6);


\draw (8,0.75) rectangle (11,4.25);
\draw (8,4.75) rectangle (11,9.25);


\node at (2.5,5) {\ldots};
\node at (7,6.5) {\ldots};
\node at (7,2) {\ldots};
\node at (12,6.5) {\ldots};
\node at (12,2) {\ldots};

\foreach \x in {0,...,3}{
  \node at (0.2,9-\x) {$u_\x$};
}
\foreach \x in {0,1}{
\node at (3.125,9+\g-2*\x) {$x_\x$};
\node at (3.125,8+\g-2*\x) {$z_\x$};
}
\foreach \x in {1,...,4}{
  \node at (0.2,\x) {$u_{n-\x}$};
}
\foreach \x in {1,2}{
\node at (3.125,-1+\g+2*\x) {$z_{\frac{n}{2}-\x}$};
\node at (3.125,\g+2*\x) {$x_{\frac{n}{2}-\x}$};
}

\node at (7,1+\g) {$z_{\frac{n}{2}-1}$};
\node at (7,3+\g) {$z_{1}$};
\node at (7,4+\g) {$z_{0}$};
\node at (7,5+\g) {$x_{\frac{n}{2}-1}$};
\node at (7,8+\g) {$x_1$};
\node at (7,9+\g) {$x_0$};

\node at (9.5,7) {$Q^{(n/2)}$};
\node at (9.5,2.5) {$Q^{(n/2)}$};

\node at (12,1+\g) {$c_{n-1}$};
\node at (12,3+\g) {$c_{\frac{n}{2}+1}$};
\node at (12,4+\g) {$c_{\frac{n}{2}}$};
\node at (12,5+\g) {$c_{\frac{n}{2}-1}$};
\node at (12,8+\g) {$c_1$};
\node at (12,9+\g) {$c_0$};

\end{tikzpicture}
\caption{Convolutional polarizing transformation $Q^{(n)}$.}
\label{fig:cpt}
\end{figure}
It is shown in \cite{morozov2018efficient} that for $n=2^m$, $\vp \in [n]$, the value of $W^{(\vp)}_n(u_0^{\vp}|y_0^{n-1})$ for CvPT can be recursively computed as
\ifonecol
\begin{align}
&W^{(2\psi)}_n(u_0^{2\psi}|y)=\sum_{w}W^{(\psi)}_{n/2}\left((u_0^{2\psi},w)X^{(2\psi+2)}\big|y'\right)
W^{(\psi)}_{n/2}\left((u_0^{2\psi},w)Z^{(2\psi+2)}|y''\right)
\label{eq:w2i}
\\
&W^{(2\psi+1)}_n(u_0^{2\psi+1}|y)=
\sum_{u_{2\psi+2},w}  W^{(\psi+1)}_{n/2}\!\!\left((u_0^{2\psi+2},w)X^{(2\psi+4)}\big|y'\right)
W^{(\psi+1)}_{n/2}\!\left((u_0^{2\psi+2},w)Z^{(2\psi+4)}\big|y''\right)
\label{eq:w2i1}
\\
&W^{(n-1)}_n(u_0^{n-1}|y)=W^{(n/2-1)}_{n/2}\left(u_0^{n-1}X^{(n)}\big|y'\right)
W^{(n/2-1)}_{n/2}\left(u_0^{n-1}Z^{(n)}\big|y''\right)
\label{eq:wn-1}
\end{align}
\else
\begin{align}
&W^{(2\psi)}_n(u_0^{2\psi}|y)=\sum_{w}W^{(\psi)}_{n/2}\left((u_0^{2\psi},w)X^{(2\psi+2)}\big|y'\right)
\nonumber \\&\times
W^{(\psi)}_{n/2}\left((u_0^{2\psi},w)Z^{(2\psi+2)}|y''\right)
\label{eq:w2i}
\\
&W^{(2\psi+1)}_n(u_0^{2\psi+1}|y)=\!\!\!
\sum_{u_{2\psi+2},w}\!\! W^{(\psi+1)}_{n/2}\!\!\left((u_0^{2\psi+2},w)X^{(2\psi+4)}\big|y'\right)
\nonumber \\&\times
W^{(\psi+1)}_{n/2}\!\left((u_0^{2\psi+2},w)Z^{(2\psi+4)}\big|y''\right)
\label{eq:w2i1}
\\
&W^{(n-1)}_n(u_0^{n-1}|y)=W^{(n/2-1)}_{n/2}\left(u_0^{n-1}X^{(n)}\big|y'\right)
\nonumber \\&\times
W^{(n/2-1)}_{n/2}\left(u_0^{n-1}Z^{(n)}\big|y''\right)
\label{eq:wn-1}
\end{align}
\fi
for $0\leq \psi<n/2-1$, where $y=y_0^{n-1}$, and $y'=y_{0}^{n/2-1}$, $y''=y_{n/2}^{n-1}$ are subvectors of $y$.
These formulae are the same as in \cite{morozov2018efficient} under permutation of the output vector $y$ by the bit-reversal permutation, which is omitted from the definition \eqref{eq:cptxz} of CvPT for the sake of simplicity.

\section{A Lower Bound on The Minimum Distance of Linear Codes}
\label{s:lower}

\subsection{Basic Definitions}

Let $\bS_n$ be the set of all linear subspaces of $\bF^n$.

Denote $a_0^{l-1} \bigcdot b_0^{l-1}=\sum_{i=0}^{l-1} a_ib_i$, where $a_i, b_i\in \bF$. For vectors $b^{(0)}, \dots, b^{(l-1)}\in\bF^t$, denote by $\la b^{(0)}, \dots, b^{(l-1)} \ra$ the linear subspace of $\bF^t$ with basis vectors $b^{(i)}$, i.e. 
\begin{align*}
\hull{b^{(0)}, \dots, b^{(l-1)}}= \set{\sum_{i=0}^{l-1}a_i b^{(i)}\left|\right. a_0^{l-1} \in \bF^l}.
\end{align*}
A sum over an empty set is assumed to be equal to zero, which implies $\hull{}= \set{\bnull^t}$, where $t$ is clear from the context.
By abuse of notation, we write $x_0x_1\dots x_{t-1}$ for $x_i\in\bF$ to denote a vector $(x_0,x_1,\dots,x_{t-1})\in\bF^t$.
\begin{example}
It can be seen that $\bS_2=\set{\hull{}, \hull{10}, \hull{01},\hull{11}, \hull{10,01}}$, and $|\bS_3|=16$.
\end{example}

\subsection{Outline of the Approach}
Consider a code $\mP$ in the form \eqref{eq:polarlike} with $\mF=[\vp]$, i.e. the set of vectors $(\bnull^\vp,u_{\vp}^{n-1})G^{(n)}$.
Code $\mP$ can be split in two sets corresponding to each value of $u_{\vp}$.
Namely, $\mP=\mP_0\cup \mP_1$, where $\mP_a$ consists of all codewords of the form $(\bnull^\vp,a,u_{\vp+1}^{n-1})G^{(n)}$.
These subsets are equal to the subsets, which probabilities are computed at the $\vp$-th phase of the SC decoding algorithm by \eqref{eq:wdef}, provided that the estimated symbols $\hat u_0^{\vp-1}$ are zero.
Since we are interested in distance properties of the code, we can assume that $\hat u_0^{\vp-1}=\bnull^\vp$.  

Let  $d_n^{(\vp)}$ be the distance between $\mP_0$ and $\mP_1$, i.e. $d_n^{(\vp)} = \min_{\dot c \in \mP_0, \ddot c \in \mP_1} \wt(\dot c + \ddot c)$.
Consider $\dot c$ and $\ddot c$, for such the minimum is achieved, i.e., $\dot c=(\bnull^\vp,0,\dot u_{\vp+1}^{n-1})G^{(n)}$, $\ddot c=(\bnull^\vp, 1,\ddot u_{\vp+1}^{n-1})G^{(n)}$, such that $d_n^{(\vp)}=\wt(\dot c+\ddot c)=\wt(\tilde c)$.
Note that $\tilde c=(\bnull^\vp, 1,\dot u_{\vp+1}^{n-1}+\ddot u_{\vp+1}^{n-1})G^{(n)}$ corresponds to value $u_{\vp}=1$, so $\tilde c \in \mP_1$. 
Hence, $d^{(\vp)}_n$ is equal to the weight of a minimum-weight codeword from $\mP_1$. In general, we can say that if $\hat u_0^{\vp-1}=u_0^{\vp-1}$, i.e., all previous symbols are estimated correctly, then the probability of erroneous estimation of $u_{\vp}$ in the case of transmission over sufficiently good binary memoryless channel is mainly defined by $d^{(\vp)}_n=\min_{c \in \mP_1}\wt(c)$. 

In section \ref{ss:gcoset} we consider the partition of $\mP$ in two sets $\mP'_0$ and $\mP'_1$ not by the value of $u_{\vp}$, but by the value of some linear combination $p_0^{j-1} \bigcdot u_{\vp}^{\vp+j-1}$ of symbols $u_{\vp}^{\vp+j-1}$.
Thus, set $\mP'_a$, $a\in\bF$ consists of all codewords $(\bnull^\vp, u_{\vp}^{n-1})G^{(n)}$ satisfying $p_0^{j-1} \bigcdot u_{\vp}^{\vp+j-1}=a$.

In section \ref{ss:recoverable} we consider  transmission of codewords through binary erasure channel (BEC) $W:\bF\to\bF\cup\set{\epsilon}$, defined as $W(x|x)=1-p_\epsilon$, $W(\epsilon|x)=p_{\epsilon}$, where $p_{\epsilon}$ is the erasure probability.
We consider  mapping of the set of erased symbols $\mE \subseteq [n]$ to the set of all linear combinations of symbols $u_{\vp}^{\vp+j-1}$, which can be recovered by the receiver by given $c_{\ove}=(c_i)_{i\notin \mE}$.
Thus, we consider a set $s\subseteq\ \bF^j$ of all vectors $p_0^{j-1} \in \bF^j$, such that the value of corresponding linear combination $p_0^{j-1}\bigcdot u_{\vp}^{\vp+j-1}$ can be recovered by receiver after erasure configuration $\mE$. It appears that $s \in \bS_j$, i.e. $s$ is a linear subspace of $\bF^j$.
 
In section \ref{ss:coseterasure}, we prove that the minimum weight of vector from $\mP'_1$ (i.e., the distance between $\mP'_0$ and $\mP'_1$) is equal to the minimum number of erasures, such that corresponding subspace $s \in \bS_j$ of coefficients of recoverable linear combinations does \textit{not} include the linear combination with coefficients $p_0^{j-1}$.

These results are combined to derive the algorithm for computing $d_n^{(\vp)}$ in the case of CvPC, which leads to the lower bound on minimum distance of CvPC and the construction of CvPS.
Furthermore, we believe that the introduced concepts and their properties can be used for other $G^{(n)}$ that have recursive structure. 

%
\subsection{Minimum Weight of Cosets and the Minimum Distance}
\label{ss:gcoset}
\begin{definition}
\label{def:coset}
Given an $n \times n$ non-singular matrix $G^{(n)}$, for a vector $p\in\bF^j$ define a generalized coset $\mC^{(\vp)}_n(p)$ as
\begin{align}
\mC^{(\vp)}_{n}(p)=\left\{u_0^{n-1}G^{(n)}| u_0^{\vp-1}=\bnull \; \wedge\; p \bigcdot u_\vp^{\vp +j-1} = 1 \right\},
\label{eq:coset}
\end{align}
\end{definition}
\begin{remark}
In the case of $j>n-\vp$, we assume in \eqref{eq:coset} that $u_l=0$ for $l\geq n$.
\label{rm:edge}
\end{remark}

%
We define \textit{the weight of the $\vp$-th bit subchannel} $W_n^{(\vp)}$ as
\begin{align*}
d^{(\vp)}_n= \min_{c \in \mC^{(\vp)}_n(1)} \wt(c).
\end{align*}
Observe that for all $j>0$ one has $\mC^{(\vp)}_n(p)=\mC^{(\vp)}_n(p,\bnull^j)$, which implies $d^{(\vp)}_n=\displaystyle\min_{c\in\mC^{(\vp)}_n(1, \bnull^j)}\wt(c)$.

\begin{lemma}
\label{lm:dcoset}
If a linear code with minimum distance $d$ is generated by rows of $G^{(n)}$ with indices from $\mI \subseteq [n]$, then 
\begin{align}
d \ge \min_{\vp\in\mI} d^{(\vp)}_n.
\label{eq:dbeta}
\end{align}
\end{lemma}
\begin{proof}
Consider the minimum-weight codeword $c_0^{n-1}=u_0^{n-1}G^{(n)}$, $\wt (c_0^{n-1})=d$.
Let $\psi$ be the first position of non-zero element in $u_0^{n-1}$.
Thus, $\psi \in \mI$, $u_\psi=1$, $u_0^{\psi-1}=\bnull$, which implies $c_0^{n-1}\in \mC^{(\psi )}_n(1)$ and $d=\wt(c_0^{n-1})\geq d^{(\psi)}_n\geq \min_{\vp\in\mI} d^{(\vp)}_n$.
\end{proof}
This bound is valid for any linear block code represented in the form of \eqref{eq:polarlike}. However, the evaluation of $d^{(\vp)}_n$ is not a simple problem for an arbitrary $G^{(n)}$. 
\subsection{Recoverable and erased vectors}
\label{ss:recoverable}

Consider transmission of  a codeword $c_0^{n-1}=u_0^{n-1}G^{(n)}$ of a code with frozen set $\mF=[\varphi]$, $u_0^{\vp-1}=\bnull
$ and dimension $k=n-\varphi$ over BEC.

The set of erased positions $\mE \subseteq [n]$ is called an \textit{erasure configuration}.
When erasure configuration $\mE$ occurs, the values $c_{\ove}=u_{\varphi}^{n-1}\hat G$ are available for the receiver, where $\hat G = G^{(n)}_{\overline{[\varphi]},\ove}$ is $k \times r$ submatrix of $G^{(n)}$ without rows from $[\vp]$ and without columns from $\mE$, $r=n-|\mE|$.
Denote by $\mU$ the set of all $\hat u_{\vp}^{n-1}$ such that $\hat u_{\vp}^{n-1}\hat G=c_{\ove}$. 
One can see  that
\begin{align}
\mU=\set{u_{\vp}^{n-1}+a_0^{k-1} \big | a_0^{k-1}\in\cs^{\perp}(\hat G)},
\label{eq:uset}
\end{align}
where for set of vectors $\mA \subseteq \bF^t$, by $\mA^\perp \subseteq \bF^t$ we denote the set of  vectors $x_0^{t-1}:\forall y_0^{t-1} \in \mA:x_0^{t-1} \bigcdot y_0^{t-1}=0$, and $\cs(A)$ is the column space of matrix $A$.
The value $u_{\vp}^{n-1}$ can be unambiguously recovered by the receiver after  erasure configuration $\mE$ iff $|\mU|=1$, i.e. $\mU=\set{u_{\vp}^{n-1}}$.

More generally, consider the recoverability of the value of a linear combination $p_{0}^{k-1}\bigcdot u_{\vp}^{n-1}$ after erasure configuration $\mE$. The set of values of  $p_{0}^{k-1}\bigcdot \hat u_{\vp}^{n-1}$ for all $\hat u_{\vp}^{n-1} \in \mU$ is given by
\begin{align}
\set{p_{0}^{k-1} \bigcdot(u_{\vp}^{n-1}+a_0^{k-1}) \;\big|\; a_0^{k-1}\in \cs^\perp(\hat G)}.
\label{eq:pset}
\end{align}
We say that vector $p_{0}^{k-1}$ is $(\mE,\varphi)$-\textit{recoverable}, if the corresponding linear combination $p_{0}^{k-1}\bigcdot u_{\vp}^{n-1}$ can be recovered unambiguously for given $c_{\ove}$, i.e., the set \eqref{eq:pset} contains only the correct value $p_{0}^{k-1} \bigcdot  u_{\vp}^{n-1}$.
Expanding the brackets in \eqref{eq:pset}, one can see that $p_{0}^{k-1}$ is $(\mE,\varphi)$-recoverable iff $\forall a_0^{k-1} \in \cs^\perp(\hat G): p_{0}^{k-1}\bigcdot a_0^{k-1}=0$, which leads to $p_0^{k-1} \in {\cs^\perp}^\perp(\hat G) =\cs(\hat G)$.
Thus, the set of $(\mE,\varphi)$-recoverable vectors is a linear space, which is equal to
$\cs(\hat G)\in \bS_k$.

\begin{definition}
\label{df:eij}
Let $s\in \bS_j$ be the space of all $p_0^{j-1}$, such that $(p_0^{j-1},\bnull^{k-j})$ is $(\mE,\vp)$-recoverable. In this case, $s$ is called a $(\mE,\vp,j)$-\textit{space} and is denoted by $\chi^{(\vp,j)}_n(\mE)$, and $\mE$ is called an $(s,\vp,j)$-\textit{configuration}.
The set of $(s,\vp,j)$-configurations is denoted by $\xi^{(\vp,j)}_n(s)$.
Thus,
\begin{align}
\chi^{(\vp,j)}_n(\mE)&=\set{p_0^{j-1}\;\big\vert\;(p_0^{j-1},\bnull^{k-j}) \in\cs\left(G^{(n)}_{\overline{[\varphi]},\ove}\right)},
\label{eq:chidef}
\\
\xi^{(\varphi,j)}_{n}(s)&= \set{\mE\;\big\vert\; \chi^{(\vp,j)}_n(\mE)=s}. \label{eq:xidef}
\end{align}
If $\mA$ is a set, denote by $2^\mA$ the set of all subsets of $\mA$.
Thus, function $\chi_n^{(\vp,j)}:2^{[n]}\to \bS_j$, maps an erasure configuration, which is a subset of $[n]$, to a linear subspace of $\bF^j$, and $\xi^{(\varphi,j)}_{n}$ returns the inverse image of $\chi_n^{(\vp,j)}$. Note that $\chi_n^{(\vp,j)}$ is not injective, so $\xi^{(\varphi,j)}_{n}:\bS_j \to 2^{2^{[n]}}$.
\end{definition}

In words, $\chi_n^{(\vp,j)}(\mE)$ defines the set of vectors $p_0^{j-1}$, for which the value of linear combination $p_0^{j-1} \bigcdot u_{\vp}^{\vp+j-1}$ can be recovered after erasure configuration $\mE$, provided that $u_0^{\vp-1}=\bnull$.
Conversely, $\xi^{(\varphi,j)}_{n}(s)$ defines the set of erasure configurations, after which the linear combination $p_0^{j-1} \bigcdot u_{\vp}^{\vp+j-1}$ can be deduced by the receiver if \textit{and only if} $p \in s$.

\begin{remark}
\label{rm:chiedge}
Let $j>k$, i.e. $j=k+h$ for some $h>0$.
In this case, the conditional part of definition \eqref{eq:chidef} is inconsistent.
We extend the definition as follows.
In Remark~\ref{rm:edge} we assume that symbols $u_{n+h}$ for $h \geq 0$ are equal to zero. Hence, these symbols are always perfectly known for the receiver, so any $\mE$ does not erase any symbol $u_{n+h}$. 
 Observe that any vector from $\bF^j\setminus \chi^{(\vp,j)}_n(\mE)$ must be \textit{not} $(\mE,\vp)$-recoverable, so for any $\mE$ and $q_0^{h-1}\in\bF^h$, we must include vector $(\bnull^k,q_0^{h-1})$ in the set $\chi^{(\vp,k+h)}_n(\mE)$. This leads to
\begin{align*}
\chi_n^{(\vp,k+h)}(\mE)=\set{(p,q) \;\big\vert\; p\in \chi_n^{(\vp,k)}(\mE),q\in\bF^ h}.
\end{align*}
Similarly, we assume that $\xi^{(\vp,k+h)}_n(s)=\emptyset$ for all $s$ which do not contain $(\bnull^k,q)$ for some $q\in\bF^h$.
\end{remark}

\begin{example} Consider $(s,0,2)$-configurations for the case of $n=2$, $c_0^1=u_0^1Q^{(2)}=(u_0+u_1,u_1)$.
For erasure configuration $\mE=\{0\}$, the only non-zero vector which is $(\mE,0)$-recoverable is $p=(0,1)$. That is, if symbol $c_0$  is erased, one can recover unambiguously only $u_1=c_1$. This means that  $\{0\} \in \xi^{(0,2)}_2(\hull{01})$. All $(s,0,2)$-configurations are
\ifonecol
\begin{align}
\xi^{(0,2)}_2(\hull{01})=\{\{0\}\}, \xi^{(0,2)}_2(\hull{10})=\emptyset, \xi^{(0,2)}_2(\hull{11})=\set{\set{1}},
\xi^{(0,2)}_2(\hull{})=\set{\set{0,1}}, \xi^{(0,2)}_2(\bF^2)=\{\emptyset\}.
\label{eq:c2}
\end{align}
\else
\begin{align}
&\xi^{(0,2)}_2(\hull{01})=\{\{0\}\}, \xi^{(0,2)}_2(\hull{10})=\emptyset, \xi^{(0,2)}_2(\hull{11})=\set{\set{1}},\nonumber\\
&\xi^{(0,2)}_2(\hull{})=\set{\set{0,1}}, \xi^{(0,2)}_2(\bF^2)=\{\emptyset\}.
\label{eq:c2}
\end{align}
\fi
That is, there are no erasure configurations, such that only $\langle 10\rangle$ (i.e. symbol $u_0$) is unambiguously recoverable, and the whole vector $u_0^1$ can be unambiguously recovered only if there are no erasures.
For the same case, the $(\mE,0,2)$-spaces are
\ifonecol
\begin{align*}
\chi_2^{(0,2)}(\emptyset)=\bF^2,\chi_2^{(0,2)}(\set{0})=\hull{01}, \chi_2^{(0,2)}(\set{1})=\hull{11}, \chi_2^{(0,2)}(\set{0,1})=\hull{}.
\end{align*}
\else
\begin{align*}
\chi_2^{(0,2)}(\emptyset)=\bF^2,\chi_2^{(0,2)}(\set{0})=\hull{01}, \\
\chi_2^{(0,2)}(\set{1})=\hull{11}, \chi_2^{(0,2)}(\set{0,1})=\hull{}.
\end{align*}
\fi
\end{example}
\begin{example}
Consider the case of $\vp=2$, $j=2$, $n=4$ and
$c_0^3=u_0^3Q^{(4)}=(u_0+u_1+u_3,u_2+u_3,u_1+u_2+u_3,u_3)$.
Since $\vp=2$ implies $u_0^1=\bnull$, one has $c_0=c_3=u_3$, $c_1=c_2=u_2+u_3$ and one can restore $u_3$ by $c_0$ or $c_3$.
Thus, $\xi^{(2,2)}_4(\hull{01})=\set{\set{1,2},\set{0,1,2},\set{1,2,3}}$. 
\end{example}
\subsection{Coset minimum weight and erasure configurations}
\label{ss:coseterasure}

For a subspace $s \in \bS_j$, we denote the minimal cardinality of $(s,\vp,j)$-configuration as
\begin{align}
\delta^{(\vp,j)}_n(s)=\min_{\mE\in\xi^{(\vp,j)}_n(s)}|\mE|,
\label{eq:deltadef}
\end{align}
assuming that the minimum over the empty set is $+\infty$.
\begin{theorem}
\label{tm:betadelta}
Let $\vp\in [n]$ and $j>0$.
For any $p \in \bF^j$,
\begin{align*}
\min_{c\in\mC^{(\vp)}_n(p)}\wt(c)=\min_{s\in \bS_j:p\notin s}\delta^{(\vp,j)}_n(s).
\end{align*}
\end{theorem}
\begin{proof}
Denote 
\ifonecol
\begin{align*}
\mA=\set{\supp(c) \big\vert c\in\mC^{(\vp)}_n(p)},\mB=\bigcup_{s\in \bS_j:p\notin s}\xi^{(\vp,j)}_n(s)=\bigcup_{s\in \bS_j:p\notin s}\set{\mE\;|\;\chi_n^{(\vp,j)}(\mE)=s}=\set{\mE \; | \; p \notin \chi_n^{(\vp,j)}(\mE)}.
\end{align*}
\else
$\mA=\set{\supp(c) \big\vert c\in\mC^{(\vp)}_n(p)}$,
\begin{align*}
\mB&=
\bigcup_{s\in \bS_j:p\notin s}\xi^{(\vp,j)}_n(s)=
\bigcup_{s\in \bS_j:p\notin s}\set{\mE\;|\;\chi_n^{(\vp,j)}(\mE)=s}
\\
&=\set{\mE \; | \; p \notin \chi_n^{(\vp,j)}(\mE)}.
\end{align*}
\fi
Then the theorem can be reformulated as
\ifonecol
$$\min_{\Omega\in \mA}|\Omega|=\min_{\mE\in \mB}|\mE|.$$
\else
$\min_{\Omega\in \mA}|\Omega|=\min_{\mE\in \mB}|\mE|.$
\fi

If $\Omega \in \mA$, then there exists $u_{\vp}^{n-1}$, such that $p \bigcdot u_{\vp}^{\vp+j-1}=1$ and  $\Omega=\supp(c_0^{n-1})$ for $c_0^{n-1}=(\bnull^\vp, u_{\vp}^{n-1})G^{(n)}$.
In this case $c_{\overline{\Omega}}=\bnull$ and the all-zero value $\hat u_{\vp}^{n-1}=\bnull$ also belongs to set \eqref{eq:uset} of possible values of $u_{\vp}^{n-1}$ for the given $c_{\overline{\Omega}}$, but $p \bigcdot \hat u_{\vp}^{\vp+j-1}=0$.
Thus, the value of  $p \bigcdot u_{\vp}^{\vp+j-1}$ is not recoverable after erasure configuration $\Omega$, which implies $p \notin \chi^{(\vp,j)}_n(\Omega)\implies\Omega \in \mB$.
So, $\Omega\in \mA \implies \Omega \in \mB$ and $\displaystyle\min_{\Omega\in \mA} |\Omega|\geq \min_{\mE\in \mB} |\mE|$. 

If $\mE \in \mB$, then $p \notin \chi^{(\vp,j)}_n(\mE)$, which by Definition~\ref{df:eij} implies $(p,\bnull^{k-j})\notin\cs(\hat G)$ and $\exists a_0^{k-1}\in \cs^\perp(\hat G): (p,\bnull^{k-j})~\bigcdot~a_0^{k-1}=~1$, which implies $p\bigcdot a_0^{j-1}=1$.
Denote $\hat c_0^{n-1}=(\bnull^\vp,a_0^{k-1})G^{(n)}$.
Since $p\bigcdot a_0^{j-1}=1$, by Definition~\ref{def:coset} one has $\hat c_0^{n-1} \in \mC^{(\vp)}_n(p)$, and therefore $\supp(\hat c)\in \mA$.
On the other hand, $\hat c_{\ove}=a_0^{k-1}\hat G=\bnull$, which means $\supp(\hat c)\subseteq \mE$.
So, $\forall \mE \in \mB \;\; \exists \Omega\in \mA: \Omega\subseteq\mE$, hence,  $\displaystyle\min_{\Omega\in \mA} |\Omega|\leq \min_{\mE\in \mB} |\mE|$.
\end{proof}
\begin{corollary}
\label{cr:betadelta}
For any $j>0:$
\ifonecol
$d^{(\vp)}_n=\min\set{\delta^{(\vp,j)}_n(s)\big\vert s \in \bS_j:(1,\bnull^{j-1}) \notin s}.$
\else
$$d^{(\vp)}_n=\min\set{\delta^{(\vp,j)}_n(s)\big\vert s \in \bS_j:(1,\bnull^{j-1}) \notin s}.$$
\fi
\end{corollary}

\section{Bound on Minimum Distance of Convolutional Polar Codes}
\label{s:mindist}
The structure of the convolutional polarizing transformation $Q^{(n)}$, $n=2^m$, enables one to compute easily $\delta^{(\vp,j)}_n(s)$,  defined in  \eqref{eq:deltadef}, for $j=3$.
By computing values of $\delta^{(\vp,3)}_n(s)$, one can obtain values of $d^{(\vp)}_n$ by Corollary~\ref{cr:betadelta} and lower bound on minimum distance by Lemma~\ref{lm:dcoset}.

Consider transmission of  $c_0^{n-1}=u_0^{n-1}Q^{(n)}$, such that $u_0^{\vp-1}=\bnull$, through BEC and let the erasure configuration be $\mE$.
The intuition behind recursive computing of $\delta^{(\vp,3)}_n(s)$ is as follows.

\textit{Consider the case of $\vp=2\psi+1<n-1$}.
Denote $x_0^{n/2-1}=u_0^{n-1}X^{(n)}$,  $z_0^{n/2-1}=u_0^{n-1}Z^{(n)}$, $\mE'=\mE\cap[\frac{n}{2}]$, $\mE''=\set{i\geq 0|i+\frac{n}{2}\in\mE}$.
Recall that $\chi^{(2\psi+1,3)}_n(\mE)$ is the set of all $p_0^2$, such that the value of $p_0^2 \bigcdot u_{2\psi+1}^{2\psi+3}$ can be deduced from $c_0^{n-1}$ after erasure configuration $\mE$.
Similarly, $\chi^{(\psi,3)}_{n/2}(\mE')$ and $\chi_{n/2}^{(\psi,3)}(\mE'')$ are the sets of $q_0^2$ and $r_0^2$, s.t. $q_0^2\bigcdot x_{\psi}^{\psi+2}$ and $r_0^2\bigcdot z_{\psi}^{\psi+2}$ are recoverable from $c_0^{n/2-1}$ and $c_{n/2}^{n-1}$ after erasure configurations $\mE'$ and $\mE''$, under assumption  $x_0^{\psi-1}=\bnull$ and $z_0^{\psi-1}=\bnull$, respectively.
By \eqref{eq:xdef}--\eqref{eq:zdef} one obtains $x_i=u_{2i}+u_{2i+1}+u_{2i+2}$ and $z_i=u_{2i+1}+u_{2i+2}$ for $i<\frac{n}{2}-1$, which, together with $u_0^{2\psi}=\bnull$, implies $x_0^{\psi-1}=z_0^{\psi-1}=\bnull$, so the above assumption holds.
Furthermore, since $u_{0}^{n-1}$ was processed by the $m$-th layer of CvPT before the transmission, the value of elements of $u_{2\psi+1}^{2\psi+3}$, as well as the value of any linear combination $p_0^{2} \bigcdot u_{2\psi+1}^{2\psi+3}$, can be deduced only from known linear combinations of elements of $x_{\psi}^{n-1}$ and $z_{\psi}^{n-1}$.
However, for any $x_{\psi+3}^{n/2-1}$, $z_{\psi+3}^{n/2-1}$ and $u_{2\psi+1}^{2\psi+3}$, one can find $u_{2\psi+4}^{n-1}$, such that $(\bnull^{2\psi+1},u_{2\psi+1}^{n-1})=\left(\bnull^\psi, x_{\psi}^{n/2-1},\bnull^\psi,z_{\psi}^{n/2-1}\right)Q^{(n)}$ as follows: set $u_{2i+2}$ to $x_{i+1}+z_{i+1}$ for $i=\frac{n}{2}-2, \ldots,\psi+1$, set $u_{n-1}$ to $z_{n/2-1}$, and set $u_{2i+1}$ to $z_i+u_{2i+2}$ for $i=\frac{n}{2}-2, \ldots,\psi+2$.
So, for any $p \in \bF^3$, even complete knowledge of $x_{\psi+3}^{n/2-1}$ and $z_{\psi+3}^{n/2-1}$  does not provide the value $p \bigcdot u_{2\psi+1}^{2\psi+3}$.
Thus, recoverable linear combinations $q_0^2\bigcdot x_{\psi}^{\psi+2}$ and $r_0^2\bigcdot z_{\psi}^{\psi+2}$ contain all information about recoverable linear combinations $p_0^2\bigcdot u_{2\psi+1}^{2\psi+3}$, and therefore $\chi^{(2\psi+1,3)}_n(\mE)$ can be uniquely deduced from given $\chi^{(\psi,3)}_{n/2}(\mE')$ and $\chi^{(\psi,3)}_{n/2}(\mE'')$.
The similar consideration for $\vp=2\psi+2$ leads to the fact that $\chi^{(2\psi+2,3)}_n(\mE)$ can also be deduced from $\chi^{(\psi,3)}_{n/2}(\mE')$ and $\chi^{(\psi,3)}_{n/2}(\mE'')$.

Let $\psi=\floor{\frac{\vp-1}{2}}$, $\bS_3=\{\mT_i\}_{i=0}^{15}$.
For any $l\in[16]$, consider $(\mT_l, \vp,3)$-erasure configuration $\mE$ for which the minimum in \eqref{eq:deltadef} is achieved, i.e. $\chi^{(\vp,3)}_n(\mE)=\mT_l$ and $|\mE|=\delta_n^{(\vp,3)}(\mT_l)$.
Obviously, $|\mE|=|\mE'|+|\mE''|$.
Let $\chi^{(\psi,3)}_{n/2}(\mE')=\mT_i$, $\chi^{(\psi,3)}_{n/2}(\mE'')=\mT_j$.
Then, $\mE'$ and $\mE''$ are also the minimum-weight $(\mT_i,\psi,3)$- and $(\mT_j,\psi,3)$- erasure configurations, respectively, i.e. $|\mE'|=\delta_{n/2}^{(\psi,3)}(\mT_i)$, and $|\mE''|=\delta_{n/2}^{(\psi,3)}(\mT_j)$.
We know that $\mT_l$ can be deduced from $\mT_i$ and $\mT_j$, i.e., for each $\vp$ and $n$ there is a function $\bfT^{(\vp)}_n(i,j)$, which returns $\mT_l$ for given  $i$ and $j$, and for considered minimum-weight $\mE$, $\mE'$, $\mE''$ one can obtain 
$\delta^{(\vp,3)}_n(\bfT^{(\vp)}_n(i,j))=\delta^{(\psi,3)}_{n/2}(\mT_i) +\delta^{(\psi,3)}_{n/2}(\mT_j)$.

It appears that $\bfT^{(\vp)}_n=\bfT^{(\vp')}_{n'}$ if $\vp \equiv\vp'\mod 2$, i.e., there are only two different functions $\bfT^{(\vp)}_n$: one for odd $\vp$ and another one for even $\vp$.
They are defined as $\bfT_o,\bfT_e:[16]\times[16]\to \bS_3$, such that
\ifonecol
\begin{align}
\bfT_o(i,j)=\set{p_0^2\big\vert \exists p'\in \mT_i,p''\in \mT_j:(p_0^2,0,0)^T=X^{(6)}_{\overline{[1]},*}p'^T + Z^{(6)}_{\overline{[1]},*}p''^T}
\label{eq:to}
\\
\bfT_e(i,j)=\left\{p_0^2\big\vert \exists p'\in \mT_i,p''\in \mT_j:(p_0^2,0)^T= X^{(6)}_{\overline{[2]},*}p'^T+Z^{(6)}_{\overline{[2]},*}p''^T \right\}.
\label{eq:te}
\end{align}
\else
\begin{align}
\bfT_o(i,j)=&\{p_0^2\big\vert \exists p'\in \mT_i,p''\in \mT_j:\nonumber\\
&(p_0^2,0,0)^T=X^{(6)}_{\overline{[1]},*}p'^T + Z^{(6)}_{\overline{[1]},*}p''^T\}
\label{eq:to}
\\
\bfT_e(i,j)=&\{p_0^2\big\vert \exists p'\in \mT_i,p''\in \mT_j:\nonumber\\
&(p_0^2,0)^T= X^{(6)}_{\overline{[2]},*}p'^T+Z^{(6)}_{\overline{[2]},*}p''^T\}.
\label{eq:te}
\end{align}
\fi
The above consideration form the following theorem.
\begin{theorem}
\label{tm:mw}
Denote  $\Delta^{(\vp)}_{n,l}=\delta^{(\vp,3)}_n(\mT_l)$ for $l\in[16]$, $n=2^m$. Then, for a CvPT
\ifonecol
\begin{align}
\Delta^{(2\psi+1)}_{n,l}&=\min_{i,j}\set{\Delta^{(\psi)}_{n/2,i}+\Delta^{(\psi)}_{n/2,j} \big\vert \bfT_o(i,j)=\mT_l}, 0\leq \psi<\frac{n}{2}
\label{eq:deltao} \\
\Delta^{(2\psi+2)}_{n,l}&=\min_{i,j}\set{\Delta^{(\psi)}_{n/2,i}+\Delta^{(\psi)}_{n/2,j} \big\vert \bfT_e(i,j)=\mT_l},-1\leq \psi<\frac{n}{2}-1
\label{eq:deltae}
\end{align}
The base of the recursion is
\begin{align}
&\delta^{(0,1)}_1(\hull{})=1, \; \delta^{(0,1)}_1(\hull{1})=0.
\label{eq:delta2}
\end{align}
\else
, for $0 \leq \psi < \frac{n}{2}:$
\allowdisplaybreaks
\begin{align}
&\Delta^{(2\psi+1)}_{n,l}=\min_{i,j}\!\set{\Delta^{(\psi)}_{n/2,i}+\Delta^{(\psi)}_{n/2,j} \big\vert \bfT_o(i,j)=\mT_l},
\label{eq:deltao} \\
&\Delta^{(2\psi)}_{n,l}=\min_{i,j}\set{\Delta^{(\psi-1)}_{n/2,i}+\Delta^{(\psi-1)}_{n/2,j} \big\vert \bfT_e(i,j)=\mT_l}.
\label{eq:deltae}
\end{align}
The base of the recursion is
\begin{align}
&\delta^{(0,1)}_1(\hull{})=1, \; \delta^{(0,1)}_1(\hull{1})=0.
\label{eq:delta2}
\end{align}
\fi
\end{theorem}

\begin{remark}
\label{rm:edgeright}
Note that formulae \eqref{eq:deltao}--\eqref{eq:deltae} include the cases of $\Delta^{(n-2)}_{n,l}=\delta^{(n-2,3)}_{n}(\mT_l)$ and $\Delta^{(n-1)}_{n,l}=\delta^{(n-1,3)}_{n}(\mT_l)$.
They can be obtained according to the assumption in Remark~\ref{rm:chiedge} as follows. For $s \in \bS_{i+h}$, denote the set of tails of length $i$ by $s\vert_i=\set{p_0^{i-1} \;\big\vert\; p_0^{i+h-1} \in s}$. We assume that any erasure configuration does not erase $u_{n-1+h}$ for any $h>0$, i.e.
\begin{align*}
\delta_n^{(n-i,i+h)}(s)=
\begin{cases}
\delta_n^{(n-i,i)}(s\vert_i), &\text{if } \forall p \in \bF^h:(\bnull^{i},p)\in s\\
+\infty, & \text{otherwise}
\end{cases}
\end{align*}
The same assumption is applied for computing the values of $\Delta^{(0)}_{1,l}=\delta^{(0,3)}_1(\mT_l)$ from the values $\delta^{(0,1)}_1(s)$ for  $s\in \bS_1$ that are given by the base \eqref{eq:delta2} of the recursion.
This assumption, though not natural since symbols $u_{n+h}$, $h\geq 0$  do not exist, allows one to employ the unified formulae \eqref{eq:deltao}--\eqref{eq:deltae} for the cases of $\vp > n-3$.
\end{remark}
\begin{remark}
\label{rm:edgeleft}
Formula \eqref{eq:deltae} in the case of $\Delta^{(0)}_{n,l}$ leads to computing $\Delta^{(-1)}_{n/2,i}=\delta^{(-1,3)}_{n/2}(\mT_i)$, which is formally equal, for a given $\mT_i$, to the minimum weight of an erasure configuration which erases values $p\bigcdot u_{-1}^2$ for and only for $p\in \mT_i$. 
For the symbols $u_{-i}$, $i>0$, we do not employ the same assumption as in Remark~\ref{rm:edgeright}.
If one assumes that symbols with negative indices are always known and employs functions $\bfT_o$ and $\bfT_e$, one would obtain that input symbols on the current layer of convolutional polarizing transformation $u_{-2}$, $u_{-1}$, and input symbol $x_{-1}=u_{-2}+u_{-1}+u_0$ on the next layer are always known, which implies that $u_0$ is always known.
This would result in incorrect value of $\Delta_{n,l}^{(0)}$.
Thus, we assume that $u_{-i}$ for $i>0$ are always erased, which leads to 
\begin{align*}
\chi_n^{(-i,j)}(\mE)=\set{(\bnull^i,p) \;\big\vert\; p\in \chi_n^{(0,j-i)}(\mE)}, 0<i\leq j.
\end{align*}
\end{remark}
\begin{proof}
The proof is in the Appendix.
\end{proof}
The values  $d^{(i)}_n$ for the case of CvPC can be computed with Algorithm~\ref{alg:convw}. 
The three-dimensional array $\tau$ of subspaces of $\bF^3$ is initialized in lines \ref{l:tinit0}--\ref{l:tinit1}, such that $\tau[0][i][j]=\bfT_e(i,j)$ and $\tau[1][i][j]=\bfT_o(i,j)$.
The values $\Delta^{(0)}_{1,*}$ are computed in lines~\ref{l:q1init0}--\ref{l:q1init1}.
Function \texttt{M1Cluster}, presented in Algorithm~\ref{alg:m1cluster}, is called to obtain $\Delta^{(-1)}_{n,*}$ for $n=1$ and $n=2^\lambda$, respectively, in lines \ref{l:m1call0} and \ref{l:m1call1}. 

The values of $\Delta^{(\vp)}_{2^\lambda,*}$ for $-1 \leq \vp < 2^\lambda$ are computed by Theorem~\ref{tm:mw} in lines \ref{l:convmain0}--\ref{l:convmain1} and stored in array $C'$, using values of $\Delta^{(\psi)}_{2^{\lambda-1},*}$ for $-1 \leq \psi < 2^{\lambda-1}$, which are stored in array $C$. The values $d_n^{(i)}$ are obtained as $d[i], i\in[n]$.

The asymptotic complexity of the Algorithm~\ref{alg:convw} is  defined by the complexity of the main loop \ref{l:convmain0}--\ref{l:convmain1}.
The complexity of the $\lambda$-th iteration of the loop is defined by the complexity of the loop in lines \ref{l:convmainmid}--\ref{l:convmainmid1}, which consists of $2^\lambda$ iterations, each of them has complexity $O(1)$. Thus, the overall asymptotic complexity is $\sum_{\lambda=1}^{\log n}O(2^\lambda)=O(n)$.
\ifonecol
\linespread{1.2}
\fi

\begin{algorithm}
\caption{Computing $d^{(i)}_n, n=2^m$ for all $i\in[n]$}
\label{alg:convw}

\DontPrintSemicolon
\KwIn{$m$}
$X\gets \begin{pmatrix}
11000\\
01110\\
00011
\end{pmatrix},
Z\gets
\begin{pmatrix}
11000\\
00110\\
00001
\end{pmatrix}
$\\
\For{$i,j =0 \ldots 15 $}{ \label{l:tinit0}
  $\tau[0\dots 1][i][j]\gets\emptyset$\\
  \For{$(p,q) \in \mT_i \times \mT_j$}{
    $r \gets pX + qZ$ \\
    \lIf{$r_3^4=\bnull^2$}{$\tau[1][i][j]\gets\tau[1][i][j] \cup r_0^2$}
    \lIf{$r_4=0$}{$\tau[0][i][j] \gets \tau[0][i][j] \cup r_1^3$} \label{l:tinit1}
  }
}
\For{$i=0\ldots 15 $}{ \label{l:q1init0}
  \lIf{$\mT_i=\bF^3$}{$C[0][0][i] \gets 0$}
  \lElseIf{$\mT_i=\hull{010,001}$}{$C[0][i] \gets 1$}
  \lElse{$C[0][i] \gets +\infty$} \label{l:q1init1}
}
$C[-1] \gets \texttt{M1Cluster}(C[0])$ \label{l:m1call0}\\
\For{$\lambda=1 \dots m$}{ \label{l:convmain0}
  $C'[0 \dots 2^{\lambda}-1][0\dots15] \gets +\infty$\\
  \For{$\vp = 0 \dots 2^{\lambda}-1$}{ \label{l:convmainmid}
    $\psi=\ceil{\frac{\vp}{2}}-1$\\
    \For{$i,j =0 \ldots 15$}{
      let $l:\mT_l=\tau[\vp \bmod 2][i][j]$\\
      $C'[\vp][l]\gets \min\set{C[\psi][i]\!+\!C[\psi][j], C'[\vp][l]}$ \label{l:convmainmid1}
    }
  }
  $C'[-1] \gets \texttt{M1Cluster}(C'[0])$ \\ \label{l:m1call1} 
  \texttt{swap}$(C,C')$ \label{l:convmain1}
} 
\lFor{$i =0\ldots n-1$}{$d[i] \gets \min_{s \in \bS^3:(1,0,0)\notin s} C[i][s]$}
\Return{$d[0\ldots n-1]$}

\end{algorithm}

\begin{algorithm}
\caption{M1Cluster}
\label{alg:m1cluster}
\KwIn{$C$, array of $\Delta^{(0)}_{n,*}$}
\KwOut{$D$, array of $\Delta^{(-1)}_{n,*}$}
$D[0 \dots 15] \gets +\infty$\\
\For{$i=0\ldots 15$}{
  \If{$\mT_i \subseteq \hull{100,010}$}{
    let $j:\mT_j=\set{(a,p_0^1) \;|\; (p_0^1,0) \in \mT_i, a \in \bF} $\\
    $D[j]\gets \min\set{C[i], D[j]}$
    }
}
\Return $D[0\dots 15]$
\end{algorithm}

\ifonecol
\linespread{2.0}
\fi

\begin{figure}
\centering
\ifonecol
\includegraphics[width=0.75\textwidth]{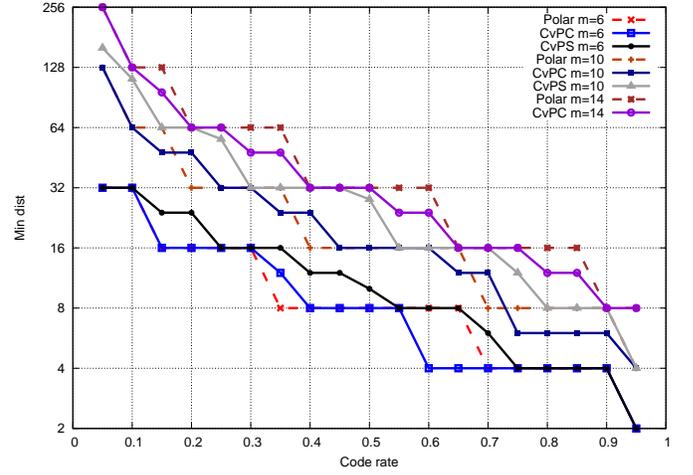}
\else
\includegraphics[width=0.5\textwidth]{d_pc_cvpc_cvps.eps}
\fi
\caption{Minimum distance of polar codes, CvPCs and CvPSs, constructed for AWGN channel for target FER $0.001$.}
\label{fig:mindist}
\end{figure}
In Fig.~\ref{fig:mindist} the lower bound on minimum distance, computed by  \eqref{eq:dbeta}, for CvPCs of lengths $64$, $1024$, $16384$ is presented. The codes are obtained via the Monte-Carlo method by minimization of the $E_b/N_0$ needed to achieve the SC decoding error probability $10^{-3}$. For comparison, we also report the results for Arikan polar codes, which are optimized in the same way.
One can see that CvPCs can have lower, equal or higher minimum distance, compared to Arikan polar codes. 

Unlike the case of Arikan polarizing transformation $A^{(n)}$, the weight of the $i$-th row of CvPT $Q^{(n)}$ is not necessarily equal to $d^{(i)}_n$. Thus, the bound \eqref{eq:dbeta} is not exact at least for codes with $\mI=\{i\}$.
In general, it is not known, for which cases the bound is exact.
However, by employing the low-weight codeword search algorithm presented in \cite{canteaut98new}, we verified that the bound is exact for CvPCs with $m=5, \dots, 13$, rates $\frac{1}{20}, \dots \frac{19}{20}$ and target FER of SC decoding $10^{-2}$, $10^{-3}$, $10^{-4}$, $10^{-5}$, and $10^{-6}$.


\section{Convolutional Polar Subcodes}
\label{s:subcodes}
In general, the SC decoding algorithm for polar-like codes does not provide ML decoding. 
The Tal-Vardy list decoding algorithm \cite{tal2015list} for polar codes can be immediately extended to the case of CvPC using the techniques presented in  \cite{morozov2018efficient}.
With sufficiently large list size  $L$ the SCL algorithm delivers near-ML decoding. The SCL decoding error probability of convolutional polar codes is lower than that of classical polar codes, but still can be improved by extending the construction of randomized polar subcodes \cite{trifonov2017randomized} to the case of convolutional polarizing transformation.

By Lemma~\ref{lm:dcoset}, any codeword $c_0^{n-1}=u_0^{n-1}G^{(n)}$ of weight $d$ corresponds to vector $u_0^{n-1}$ with at least one symbol $u_i=1, i \in \mI$, such that $d^{(i)}_{n}\leq d$.
In the case of polar codes, $d^{(i)}_n$ is equal to the weight of the $i$-th row of $A^{(n)}$. In the case of CvPCs one can obtain $d^{(i)}_n$ by Algorithm~\ref{alg:convw}.

A code construction, which has low SCL decoding error probability, was proposed in \cite{trifonov2017randomized} for the case of classical polar codes as polar subcodes.
Polar subcodes are obtained as a generalization of polar codes, where some symbols $u_\vp, \vp\in\mD$, called dynamic frozen symbols, are not set to zero, but to  linear combinations of previous symbols $u_i,i<\vp$. This approach can be immediately extended to the case of convolutional polarizing transformation. Namely, the dynamic freezing constraints should be constructed, so that they involve all non-frozen symbols $u_i$ with the smallest $d_n^{(i)}$, but the indices of dynamic frozen symbols $i\in\mathcal D$ should be as small as possible, so that the SCL decoding algorithm can process these constraints at the earliest possible phases, minimizing thus the probability of a correct path being killed.

This results in the following code construction algorithm: 
\begin{enumerate}
\item Construct $(n,k+f)$ convolutional polar code, i.e. assign $u_{\mS}= \bnull$ for the static frozen set $\mS \subset [n]$ of worst $n-k-f$ bit subchannels.

\item Choose dynamic frozen set $\mD\subseteq [n] \setminus \mS$ as the set of $f$ indices of minimum-weight bit subchannels with the largest indices, that are not static frozen.
Set
\begin{align*}
u_i= \sum_{j\in \mI} V_{i,j}u_j, \; i \in \mD,\; \mI=[n]\setminus \mF,
\end{align*}
where the frozen set $\mF=\mS \cup \mD$ consists of indices of static frozen or dynamic frozen symbols, and $V_{i,j}$ are distributed uniformly over $\bF$.
\end{enumerate}
The set $\mI$ for a convolutional polar code optimized for SC decoding can be chosen either by evolution of erasure probabilities proposed in \cite{ferris2017convolutional}, or by Monte-Carlo simulations of genie-aided SC decoder. 
Due to lack of analysis techniques for the list SC decoding algorithm, the optimal value of $f$ should be determined by simulations.

Another component of the construction introduced in \cite{trifonov2017randomized} is type-B dynamic freezing constraints, which are imposed on the symbols transmitted over the least reliable yet unfrozen subchannels. These  constraints speed up error propagation for incorrect paths in the list SC algorithm, so that the probabilities \eqref{eq:wdef} of these paths decrease quickly, reducing thus the probability of a correct path being killed. 
However, simulations of moderate-length CvPS show that type-B dynamic frozen symbols do not provide any noticeable gain in the case of CvPS.

\section{Performance of Convolutional Polar Subcodes}
\label{s:perf}

\begin{figure}
\centering
\ifonecol
\includegraphics[width=0.75\textwidth]{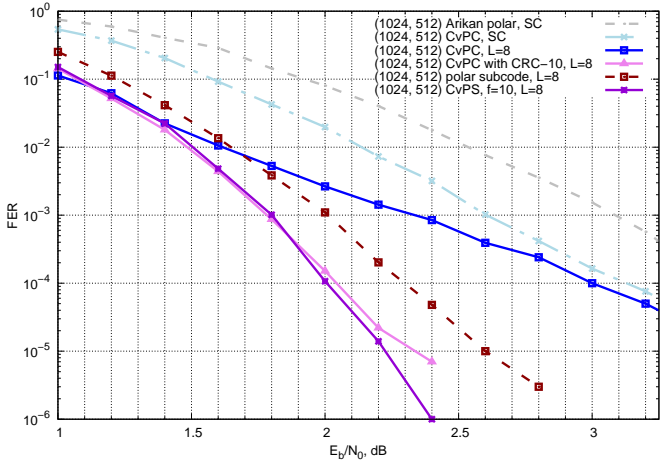}
\else
\includegraphics[width=0.5\textwidth]{subcode.eps}
\fi
\caption{Performance of $(1024, 512)$ CvPS with $f=10$ in AWGN channel}
\label{fig:1024}
\end{figure}

In Fig.~\ref{fig:1024} the performance of $(1024,512)$ CvPS, polar code and polar subcode is presented for $f=10$ for the case of AWGN channel. The polar code and the polar subcode are constructed for AWGN channel with $E_b/N_0=2$ dB using Gaussian approximation of density evolution \cite{trifonov2012efficient}, and the CvPS is constructed for the same channel using Monte-Carlo simulations for subchannels qualities.
One can see that the CvPS outperforms randomized polar subcodes \cite{trifonov2017randomized}, CvPC \cite{ferris2017convolutional} and CvPC concatenated with CRC-10.

\begin{figure}
\centering
\ifonecol
\includegraphics[width=0.75\textwidth]{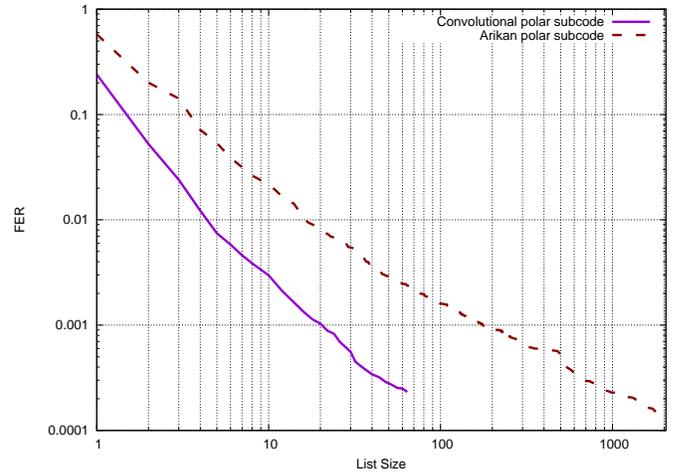}
\else
\includegraphics[width=0.5\textwidth]{4096_conv_list_fer.eps}
\fi
\caption{Performance of $(4096, 2048)$ CvPS with $f=12$ in AWGN channel with $E_b/N_0=1.25$ dB}
\label{fig:4096list}
\end{figure}

In Fig.~\ref{fig:4096list} the performance of a $(4096,2048)$ CvPS with $f=12$  type-A dynamic frozen symbols is presented. Transmission of BPSK-modulated symbols over AWGN channel with $E_b/N_0=1.25$ dB is considered. The decoding algorithm is the SCL decoding with different values of list size that are shown in the x-axis.
The performance of CvPSs is compared to that of a polar subcode with $f=12$ type-A dynamic frozen symbols and $52$ type-B dynamic frozen symbols. One can see that the CvPS under SCL decoding with the same list size outperforms classical polar subcodes.
The smaller list size can be used to achieve the same FER, which allows less sophisticated hardware implementation.

\begin{figure}
\centering
\ifonecol
\includegraphics[width=0.75\textwidth]{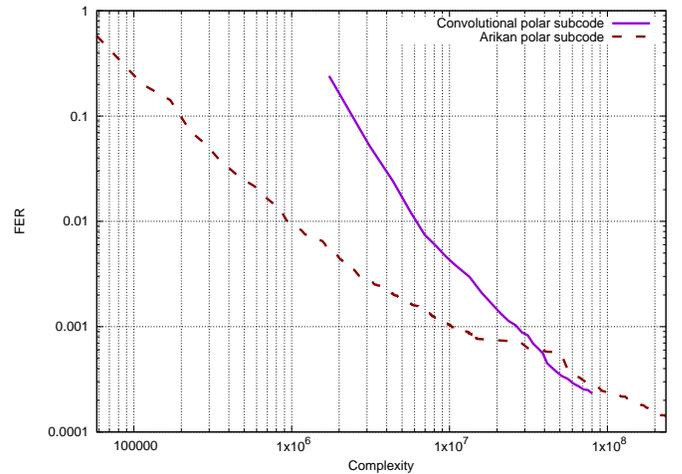}
\else
\includegraphics[width=0.5\textwidth]{4096_conv_compl_fer.eps}
\fi
\caption{SCL decoding complexity of $(4096, 2048)$ CvPS with $f=12$}
\label{fig:4096compl}
\end{figure}

In Fig.~\ref{fig:4096compl} the complexity (the number of operations) of SCL decoding, based on the expressions derived in \cite{morozov2018efficient}, of the described above codes is compared for list size $L=1\dots 64$ for the CvPS and $L=1\dots 1024$ for the polar subcode.
The complexity is obtained as the number of additions and comparisons of LLRs.
The complexity of SC decoding for CvPS is approximately $46.5n\log n$, as shown in \cite{morozov2018efficient}. The complexity of SC decoding of polar codes is $n \log n$.
However, as was shown in \cite{ferris2017convolutional}, CvPT induces stronger polarization than Arikan polarizing transformation, so the smaller list size is needed to achieve the same FER. This leads to the smaller complexity needed to achieve FER less than $6\cdot 10^{-4}$ in the case of CvPS, because achieving this FER requires list size $L=352$ for polar subcodes and only $L=28$ for CvPS. 
Furthermore, for a large list size the SCL decoding is near-ML, and for sufficiently good channel FER of ML-decoding is mainly defined by the minimum distance and the error coefficient. Dynamic frozen symbols decrease the error coefficient and may even increase the minimum distance of a CvPS.
In Fig.~2 one can see that the minimum distance of CvPS is higher than that of CvPC.

\section{Conclusions}
In this paper a tight lower bound on minimum distance of convolutional polar codes is provided.
Furthermore, a generalization of the randomized construction of polar subcodes to the case of convolutional polarizing transformation is proposed.
Simulations show that the proposed code construction has lower frame error rate under SCL decoding \cite{morozov2018efficient} compared to polar subcodes with the same list size.
The complexity for  achieving the same FER with convolutional polar subcodes can be lower than in the case of polar subcodes  \cite{trifonov2017randomized} based on Arikan polarizing transformation.


\appendix

\textit{Proof of Theorem~\ref{tm:mw}.}
For erasure configuration $\mE \subseteq [n]$, denote $\mE'=\mE \cap [n/2]$ and $\mE''=\set{j-n/2 \mid j\in\mE \setminus [n/2]}$.
We now consider the case of $\varphi=2\psi+1$ and prove \eqref{eq:deltao}.

Note that $u_0^{2\psi}=\bnull^{2\psi+1}$ implies $x_0^{\psi-1}=z_0^{\psi-1}=\bnull^\psi$.
By \eqref{eq:cptxz} one obtains 
\begin{align*}
\hat Q=\left(\hat X \hat Q'\;,\;
\hat Z \hat Q''\right),
\end{align*}
where $\hat Q=Q^{(n)}_{\overline{[2\psi+1]},\ove}$,
$\hat Q'=Q^{(n/2)}_{\overline{[\psi]}, \overline{\mE'}}$, 
$\hat Q''=Q^{(n/2)}_{\overline{[\psi]}, \overline{\mE''}}$,
$\hat X=X^{(n)}_{\overline{[2\psi+1]},\overline{[\psi]}}$,
$\hat Z=Z^{(n)}_{\overline{[2\psi+1]},\overline{[\psi]}}$.
By \eqref{eq:chidef}, 
\ifonecol
\begin{align*}
p_0^2\in \chi^{(\vp,3)}_n(\mE) \iff \exists q:(p_0^2,\bnull^{k-3})^T=\hat Qq^T=(\hat X \hat Q', \hat Z \hat Q'')q^T
=\hat X \hat Q'q'^T + \hat Z \hat Q''q''^T,
\end{align*}
\else
$p_0^2\in \chi^{(\vp,3)}_n(\mE)$ iff there exists $q$:
\begin{align*}
(p_0^2,\bnull^{k-3})^T=\hat Qq^T=(\hat X \hat Q', \hat Z \hat Q'')q^T
=\hat X \hat Q'q'^T + \hat Z \hat Q''q''^T,
\end{align*}
\fi
where $q=(q',q'')$, $k=n-\vp$, which implies, in particular,
\begin{align}
(\hat X\hat Q'q'^T)_{\overline{[3]}} = (\hat Z\hat Q''q''^T)_{\overline{[3]}}.
\label{eq:o3}
\end{align}
Denote $a=q'\hat Q'^T$, $b=q''\hat Q''^T$. Thus, $a\in\cs(\hat Q')$, $b \in \cs(\hat Q'')$.
Then \eqref{eq:o3} implies $a\hat X_{\overline{[3]},*}^T = b\hat Z^T_{\overline{[3]},*}$, so from \eqref{eq:xdef}--\eqref{eq:zdef} one obtains
\begin{align*}
a\begin{pmatrix}
000000\ldots\\
100000\ldots\\
111000\ldots\\
001110\ldots\\
\dots
\end{pmatrix}=
b\begin{pmatrix}
000000\ldots\\
100000\ldots\\
011000\ldots\\
000110\ldots\\
\dots
\end{pmatrix},
\end{align*}
which leads to the system of equations 
\begin{align}
\begin{cases}
a_i+a_{i+1}=b_i, & i=1 \ldots n/2-\psi-2\\
a_i=b_i, & i=2 \ldots n/2-\psi-1
\end{cases}
\label{eq:xizeta}
\end{align}
It is easy to see that \eqref{eq:xizeta} implies $a_i=b_i=0$ for $i\geq 3$.
Let $k'=n/2-\psi$. By above consideration, for any $p\in\bF^3$ one has $(p, \bnull^{k-3})\in \cs(\hat Q)$ iff 
\ifonecol
\begin{align}
\exists p', p'' \in \bF^3: 
&(p', \bnull^{k'-3})\in \cs(\hat Q'), (p'', \bnull^{k'-3})\in \cs(\hat Q'') \text{ and } (p,\bnull^3)^T=
\begin{pmatrix}
100\\
110\\
010\\
011\\
001\\
001
\end{pmatrix}(p')^T
+
\begin{pmatrix}
100\\
100\\
010\\
010\\
001\\
001
\end{pmatrix}(p'')^T.
\label{eq:ppspss}
\end{align}
\else
there exists $p', p'' \in \bF^3$, s.t. $(p', \bnull^{k'-3})\in \cs(\hat Q')$, $(p'', \bnull^{k'-3})\in \cs(\hat Q'')$, and
\begin{align}
&(p,\bnull^3)^T=
\begin{pmatrix}
100\\
110\\
010\\
011\\
001\\
001
\end{pmatrix}(p')^T
+
\begin{pmatrix}
100\\
100\\
010\\
010\\
001\\
001
\end{pmatrix}(p'')^T.
\label{eq:ppspss}
\end{align}
\fi
Note that two last elements of vector in the left-hand side equals $0$, and two last rows in the right hand size of  \eqref{eq:ppspss} are identical, so last rows of these matrices can be removed. The resulting matrices are equal to 
$X^{(6)}_{\overline{[1]},*}$ and $Z^{(6)}_{\overline{[1]},*}$, respectively. Recalling \eqref{eq:chidef}, one obtains that $\chi_n^{(2\psi+1,3)}(\mE)$ consists of all $p_0^2$, for which there exist ${p'}\in \chi_{n/2}^{(\psi,3)}(\mE')$, $p'' \in \chi_{n/2}^{(\psi,3)}(\mE'')$:
\begin{align}
(p_0^2,\bnull^2)^T=X^{(6)}_{\overline{[1]},*}p'^T+Z^{(6)}_{\overline{[1]},*}p''^T.
\label{eq:xpszpss}
\end{align}
Observe that \eqref{eq:xpszpss} is equivalent to the equation in the right-hand side of \eqref{eq:to}.
Obviously, $|\mE|=|\mE'|+|\mE''|$ and the minimal cardinality of $(\mT_l,2\psi+1,3)$-configuration  $|\mE|$ for each $\mT_l \in \bS_3$ can be found exactly as it is stated in \eqref{eq:deltao}.

Equality \eqref{eq:deltae} can be proved similarly.

\bibliographystyle{IEEETran}

\end{document}